\documentclass[prb,aps,floatfix,amsmath,amssymb,superscriptaddress,tightenlines,twocolumn]{revtex4}
\usepackage{graphicx}
\usepackage{epstopdf}
\newcommand{\be}{\begin{equation}}
\newcommand{\ee}{\end{equation}}
\usepackage{amsmath}
\usepackage{amsfonts}
\usepackage{bm}
\usepackage{color}
\usepackage{subfigure}
\usepackage{amsthm}
\usepackage[latin1]{inputenc}
\usepackage{tikz}
\usepackage{comment}

\newcommand{\bra}[1]{\left\langle #1 \right|}
\newcommand{\ket}[1]{\left|#1\right\rangle}

\newcommand{\Tr}{\textrm{Tr}}

\newcommand{\overbar}[1]{\mkern 1.5mu\overline{\mkern-1.5mu#1\mkern-1.5mu}\mkern 1.5mu}

\newtheorem{lem}{Lemma}
\newtheorem{thm}{Theorem}
\newtheorem{defi}{Definition}

\begin{document}
\title{Local integrals of motion and the logarithmic lightcone in many-body localized systems}

\author{Isaac H. Kim}
\affiliation{Perimeter Institute for Theoretical Physics, Waterloo ON N2L 2Y5, Canada}
\author{Anushya Chandran}
\affiliation{Perimeter Institute for Theoretical Physics, Waterloo ON N2L 2Y5, Canada}
\author{Dmitry A. Abanin}
\affiliation{Perimeter Institute for Theoretical Physics, Waterloo ON N2L 2Y5, Canada}
\affiliation{Institute for Quantum Computing, University of Waterloo, Waterloo ON N2L 3G1, Canada}

\date{\today}
\begin{abstract}
We propose to define full many-body localization in terms of the recently introduced integrals of motion[Chandran et al., arXiv:1407.8480], which characterize the time-averaged response of the system to a local perturbation. The quasi-locality of such integrals of motion implies an effective lightcone that grows logarithmically in time. This subsequently implies that (i) the average entanglement entropy can grow at most logarithmically in time for a global quench from a product state, and (ii) with high probability, the time evolution of a local operator for a time interval $|t|$ can be classically simulated with a resource that scales polynomially in $|t|$.

\end{abstract}

\maketitle

\section{Introduction}
One of the most striking phenomena that occurs in disordered quantum non-interacting systems is Anderson localization, a purely quantum-mechanical mechanism whereby a single particle is localized indefinitely in some region in space.\cite{Anderson1958} There has been a recent surge of interest\cite{Gornyi2005,Basko2006,Prosen2008,Pal2010,Bardarson2012,Vosk2013,Serbyn2013-1,Serbyn2013,Huse2013,Huse2013-1,Bauer2013,Pekker2014,Kjall2014} in studying the effects of interaction in systems that exhibit Anderson localization. Research in this direction was largely inspired by Refs.~[\onlinecite{Gornyi2005,Basko2006}], in which it was  perturbatively argued that localization can survive at finite energy density. In particular, localization property of infinite-temperature states in disordered spin chains was studied by Pal and Huse\cite{Pal2010}, who numerically observed a transition of the energy level statistics between the localized and the delocalized phase. Other numerical studies\cite{Prosen2008,Bardarson2012} found that the localized phase exhibits a logarithmic growth of entanglement entropy under a quantum quench, which is different from the linear growth that is observed in delocalized systems.\cite{Calabrese2005,Calabrese2006,Kim2013a} These properties are strongly believed to be universal characteristics of many-body localized (MBL) systems.

To explain these observations, in particular, the logarithmic spreading of entanglement, it was proposed that MBL phase is characterized by an extensive set of emergent local integrals of motion.\cite{Serbyn2013,Huse2013} The existence of the complete set of integrals of motion introduced in Ref.[\onlinecite{Serbyn2013,Huse2013}] directly implies the numerically observed properties of the MBL phase, but a systematic method for finding such integrals of motion remains to be found.

Recently, we introduced a systematic method for finding an alternative set of local integrals of motion in the MBL phase\cite{Chandran2014} (see also Ref.[\onlinecite{Ros2014}] for a related work). The integrals of motion of Ref.[\onlinecite{Chandran2014}] can be constructed systematically, and have a clear physical meaning, describing the long-time response of the MBL system to a local perturbation. However, it is currently not known as to whether the properties of the MBL phase follow from the locality of these integrals of motion.

The current state of affairs raises a natural question. Namely, does the locality of the integrals of motion defined in Ref.[\onlinecite{Chandran2014}] directly imply the numerically observed properties of the MBL phase? Moreover, can the locality of these integrals of motion be viewed as a defining property of the MBL phase? Here we answer these questions in the affirmative. Under a mild locality condition, we prove (i) a statement that suggests an absence of transport and (ii) the slow growth of average entanglement entropy under a quantum quench. We also prove a fact that may be relevant in numerically simulating these systems: that the long-time dynamics of the majority of the disorder realizations can be efficiently simulated by a classical computer.

The key result that underlies these properties is a variant of the Lieb-Robinson bound that is tailor-made to incorporate the locality of the integrals of motion (Theorem 1 below). Lieb-Robinson bound states that there exists an effective speed of light in an interacting quantum many-body system with local interactions, such that correlations decay exponentially outside the effective light cone.\cite{LR1972} It was first shown by Burrell and Osborne that a much stronger bound can be achieved for localized noninteracting systems,\cite{Burrell2007} and the bound was subsequently tightened by Hamza et al.\cite{Hamza2011} We prove a variant of these bounds under the locality condition imposed on the integrals of motion. The key feature of our bound is that it exhibits an effective logarithmic lightcone, and as such, the effective speed of light asymptotically approaches to $0$ in the infinite time limit. Indeed, this clearly demonstrates that the integrals of motion defined in Ref.[\onlinecite{Chandran2014}] correctly capture the universal signatures of the MBL phase. We therefore propose that the exponential decay of such integrals of motion can be adopted as a definition of many-body localization.

Another important issue that we address in this paper is whether the integrals of motion are indeed local. While proving such a statement remains as an open problem, there are certain encouraging facts that support the prospect of our approach. Most importantly, the locality condition we impose on the integrals of motion (Definition 1) is based on its \emph{average} decay property. Working with averaged quantities brings certain advantages compared to an approach based on probability estimates. Due to the existence of rare resonant regions that occur in the MBL phase,\cite{Bauer2013} approaches based on probabilistic estimates must account for the density of such regions. While this is by no means impossible,\cite{Imbrie2014} there is little doubt that the existence of such regions poses a great technical challenge in rigorously proving the existence of the MBL phase. In contrast, our approach deliberately sidesteps this issue by simply focusing on the average behavior. It should be also noted that the strength of the integrals of motion are indeed numerically observed to decay exponentially on average.

The rest of the paper can be largely divided into two parts. In the first part, we formulate a precise definition of the locality of the integrals of motion, and argue its validity. The second part is devoted to the implication of this definition, beginning with the proof of the Lieb-Robinson bound that exhibits the logarithmic lightcone. The integrals of motion and the precise sense in which they are local is defined in Section \ref{section:IM}.  In Section \ref{section:ZVLR}, we prove the aforementioned variant of the Lieb-Robinson bound. In Section \ref{section:Implications}, we prove the immediate consequences of the bound.

\section{Integrals of motion and localization length\label{section:IM}}

Following Ref.[\onlinecite{Chandran2014}], we introduce a set of integrals of motion that can be defined on any locally interacting Hamiltonian. Without loss of generality, we assume the Hamiltonian has the following form:
\begin{equation}
H=\sum_{j=1}^N h_j, \nonumber
\end{equation}
where $h_j$ is a local term that has a bounded support and strength, and $N$ is the total number of particles, e.g., spins or fermions. We use the infinite-time average of $e^{iHt}h_je^{-iHt}$ as the integrals of motion, which we denote as $\tilde{h}_j$.
\begin{equation}\label{eq:hj}
\tilde{h}_j= \lim_{T \to \infty} \frac{\int^{T}_{-T} e^{iHt} h_je^{-iHt} dt}{2T}
\end{equation}
We note that the infinite-time average of {\it any} local operator is an integral of motion, but for our purposes we focus on the particular integrals of motion (\ref{eq:hj}).

There are two important properties of $\tilde{h}_j$ that will be exploited extensively later. First, $H$ can be expressed as a sum of $\tilde{h}_j$:
\begin{equation}
H=\sum_{j=1}^{N} \tilde{h}_j.
\end{equation}
 Second, $\tilde{h}_j$ commutes with each other:
\begin{equation}
[\tilde{h}_i, \tilde{h}_j]=0.
\end{equation}
These two facts follow directly from the definition of $\tilde{h}_j$, and as such, they form a set of integrals of motion for any quantum many-body systems.

In Ref.[\onlinecite{Chandran2014}], it was observed that the strength of $\tilde{h}_j$ decays exponentially on average. In particular, the inverse of their average spatial decay rate was related to the localization length. This motivates a formal definition of what it means for a system to be localized, which we describe below; we use a notational convention that $\| O \|$ denotes the operator norm of $O$, i.e., the largest modulus of the eigenvalues of $O$. We denote the average over the distribution $\mu$ as $\mathbb{E}_{ \mu}$.
\begin{defi}
A family of Hamiltonians $H_{\mu}=\sum_j h_j(\mu)$ over a disorder distribution $ \mu$ is localized at a lengthscale $\xi$ if
\begin{equation}
\mathbb{E}_{\mu}\|[\tilde{h}_j(\mu), O]\| \leq e^{-\frac{x}{\xi}} \|O \| ,
\end{equation}
for all $j$ and $O$, where the support of $O$ is distance $x$ away from $j$.
\end{defi}
We would like to caution the readers by emphasizing that $\xi$ is \emph{not} necessarily equal to the localization length. This is due to a simple reason that, if a system is localized at a lengthscale $\xi$, it is localized at a lengthscale $\xi'$ as well, where $\xi'>\xi$. On this ground, $\xi$ can be thought as an upper bound on the localization length.

\subsection{Noninteracting case}
In the preceding definition, $h_j(\mu)$ is a geometrically local Hamiltonian which is randomly sampled from a probability distribution $\mu$. As an example, consider a model describing a noninteracting Anderson insulator in one spatial dimension:
\begin{equation}
H_{ \mu} = \sum_{j=1}^{N}-t(a^{\dagger}_j a_{j+1} + a^{\dagger}_{j+1} a_j) +  W_jn_j, n_j=a^{\dagger}_ja_j \label{eq:AndersonModel}
\end{equation}
where $W_j$ is uniformly distributed in $[-W,W]$. In this case, one possible choice of $h_j(\mu)$ would be $-t(a^{\dagger}_j a_{j+1} + a^{\dagger}_{j+1} a_j) + W_j n_j$, where $W_j$ is a random variable sampled from the $j$th component of a multivariate i.i.d. distribution $\mu=(W_1,W_2, \cdots)$. We note in passing that there might be other legitimate choices of $\tilde{h}_j(\mu)$, although we do not expect the rest of this paper to depend on such details.

We have studied the decay property of the integrals of motion for the above model, by exploiting the fact that the system is noninteracting. This implies that $\tilde{h}_j$ can be written as a sum of terms that are quadratic in terms of the creation and the annihilation operators. It is convenient to write down the Hamiltonian in terms of the Majorana operators($c_n$), which are related to the creation($a_n^{\dagger}$) and annihilation($a_n$) operators as follows:
\begin{equation}
c_{2n-1} =a_n^{\dagger} + a_n \quad c_{2j}=\frac{a_n-a_n^{\dagger}}{i}. \nonumber
\end{equation}
Without loss of generality, any noninteracting Hamiltonian can be written as follows:
\begin{equation}
\tilde{h}_j = \sum_{n,m} \frac{i\tilde{h}_{nm}^j}{2} c_n c_m. \nonumber
\end{equation}
The commutator between $\tilde{h}_j$ and an operator $O$ can be bounded by the triangle inequality:
\begin{equation}
\|[\tilde{h}_j, O] \| \leq \|[\tilde{h}_j(x),O] \| + \|[\tilde{h}_j - \tilde{h}_j(x), O] \|,
\end{equation}
where
\begin{equation}
\tilde{h}_j(x) = \sum_{|n-j|, |m-j| < x} \frac{i\tilde{h}_{nm}^j}{2} c_n c_m. \nonumber
\end{equation}
It should be noted that $[\tilde{h}_j(x),O]=0$, since the support of $\tilde{h}_j(x)$ does not overlap with the support of $O$. Now we have an upper bound on $\|[\tilde{h}_j, O] \|$:
\begin{equation}
\|[\tilde{h}_j, O] \|\leq 2\eta(x) \| O \|, \label{eq:eta_definition}
\end{equation}
where
\begin{equation}
\eta_1(x)=\mathbb{E}_{\mu}\sum_{\min(|n-j|,|m-j|) \geq x } |h_{nm}|.
\end{equation}
Since $\eta_1(x)$ is independent of $O$, it serves as a useful figure of merit for assessing whether the given system is localized or not. If $\eta_1(x)$ decays exponentially in $x$, Definition 1 implies that the system is localized at some finite lengthscale. This is indeed what we observe in FIG.\ref{fig:IoMDecay_noninteracting}. The plot suggests that Definition 1 clears the minimal requirement: that it is applicable to a noninteracting localized system.
\begin{figure}[h]
\includegraphics[width=3.3in]{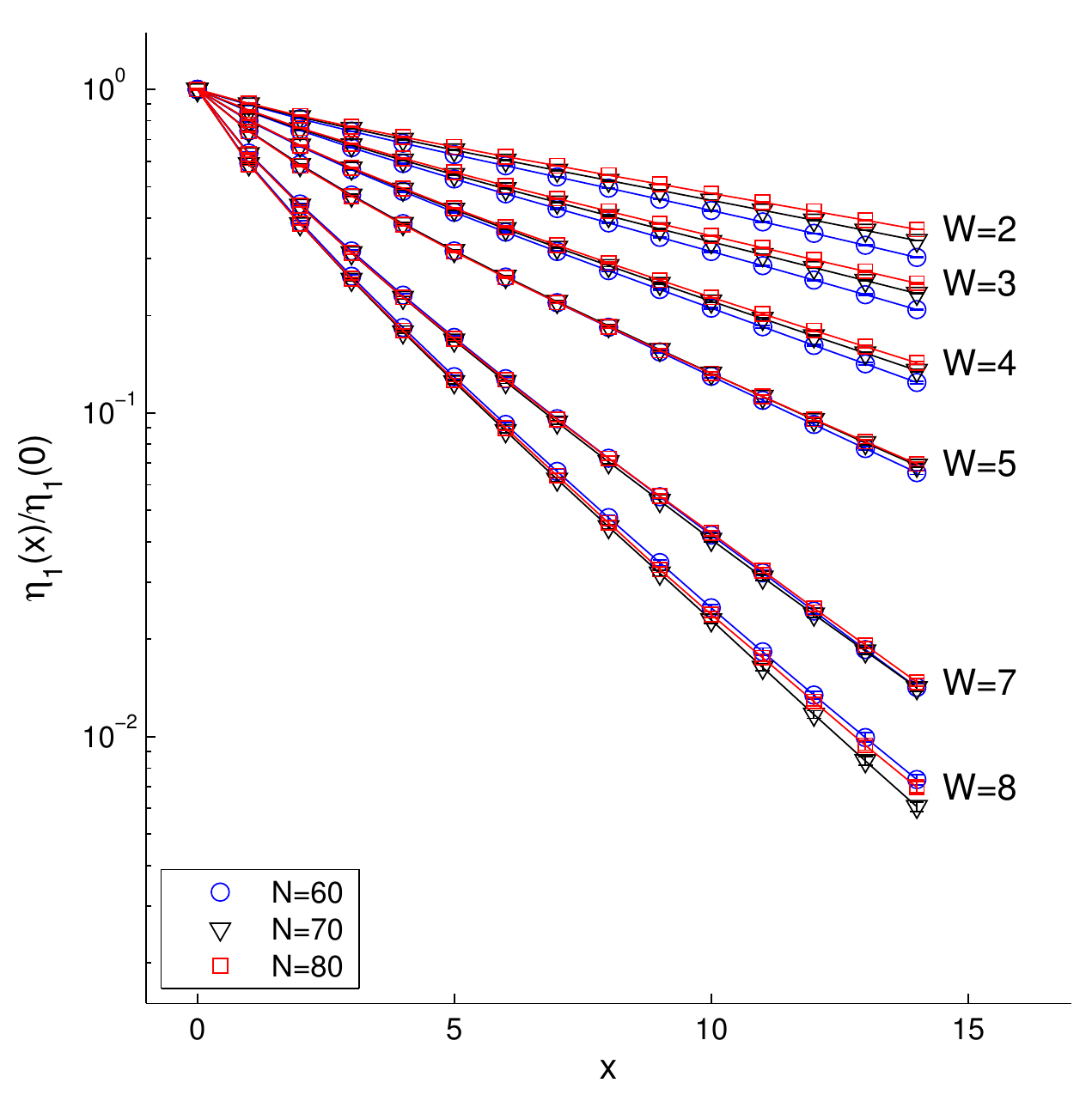}
\caption{Average decay property of the integrals of motion for noninteracting systems(Eq.\ref{eq:AndersonModel}). The system consists of $N=60,70,$ and $80$ noninteracting spinless fermions. The blue, red, and the black plots correspond to $N=60,70,$ and $80$. The average is taken over $1\times 10^4$ disorder realizations. We studied the decay property of the integrals of motion $\tilde{h}_{30}$. The variation of $\eta_1(x)$ between different system sizes were within $\sim10\%$ of their average values. The decay profiles are consistent with the exponential decay. Due to the small localization length in weakly disordered systems($W=2,3,4$), different system sizes lead to slightly different decay rate. However, this inconsistencies quickly vanish as the disorder strength is increased.  \label{fig:IoMDecay_noninteracting}}
\end{figure}

\subsection{Interacting case}
Definition 1 is well-defined for any locally interacting quantum many-body system, and as such, it is applicable even if interaction is added to Eq.(\ref{eq:AndersonModel}). Indeed, we apply it to the XXZ spin chain with nearest-neighbor interactions and random magnetic field in the $z$ direction, which can be mapped to such a model under Jordan-Wigner transformation. The Hamiltonian in the spin basis is as follows:
\begin{equation}
H=J_x\sum_{j=1}^{N}(\sigma_j^x \sigma_{j+1}^{x}+\sigma_j^y \sigma_{j+1}^y) + J_z\sum_{j=1}^{N} \sigma_j^z \sigma_{j+1}^z + \sum_{j=1}^{N} W_j \sigma_j^z,\label{eq:interacting_Hamiltonian}
\end{equation}
where $W_j$ is randomly distributed in an interval $[-W,W]$. For numerical reasons, here we use a different yet related norm as a figure of merit: the Frobenius norm. Recall that the Frobenius norm of an operator $A$ is defined as $\Tr(A^{\dagger}A)$. It is straightforward to verify
\begin{equation}
\| A \|^2 \leq \Tr(A^{\dagger} A).
\end{equation}
Therefore, if the Frobenius norm decays sufficiently fast, it would imply that the operator norm of $\| A \|$ decays fast as well. As we did in the noninteracting case, one can conservatively estimate the commutator between $\tilde{h}_j$ and $O$ by approximating $\tilde{h}_j$ by an operator with a restricted support, which is again denoted as $\tilde{h}_j(x)$. Without loss of generality, one can express any operator as  a sum of orthonormal basis set that spans the set of operators. A canonical choice is the generalized Pauli-operators, which is of the following form:
\begin{equation}
\sigma = 2^{-\frac{N}{2}}\sigma^{i_1} \otimes \sigma^{i_2} \otimes \cdots \otimes \sigma^{i_n},
\end{equation}
where $\sigma^i$ is an element of the Pauli group and $2^{-\frac{N}{2}}$ is the suitable normalization factor to ensure $\Tr(\sigma^{\dagger} \sigma)=1$.

Analogous to our choice for the noninteracting case, we attempt to approximate $\tilde{h}_j$ by the following operator:
\begin{equation}
\tilde{h}_j(x)= \frac{1}{d_{\overline{j(x)}}}\Tr_{\overline{j(x)}} \tilde{h}_j,
\end{equation}
where $j(x)$ is a set of sites that are distance $x$ or less away from $j,$ $\overline{j(x)}$ is its complement, and $d_{\overline{j(x)}}$ is the dimension of the Hilbert space associated to $d_{\overline{j(x)}}$. Since all the elements in the Pauli group aside from the identity matrix is traceless, $\tilde{h}_j(x)$ only consists of terms that act trivially on $\bar{j}(x)$. In particular, it commutes with any operator $O$ whose support is distance $x$ or more away from $j$. The key identity is the following:
\begin{equation}
\Tr(\tilde{h}_j^2)=\Tr(\tilde{h}_{j}(x)^2) + \Tr((\tilde{h}_j - \tilde{h}_j(x))^2), \label{eq:Frobenius_Identity}
\end{equation}
which can be verified easily by expanding all the terms. This leads to the following sequence of inequalities:
\begin{align}
\mathbb{E}_{\mu}\|[\tilde{h}_j,O]\| &= \mathbb{E}_{\mu}\|[\tilde{h}_j - \tilde{h}_j(x),O]\| \nonumber \\
&\leq 2\mathbb{E}_{\mu} \| \tilde{h}_j-\tilde{h}_j(x)\| \| O\| \nonumber \\
&\leq 2\mathbb{E}_{\mu} (\Tr( \tilde{h}_j-\tilde{h}_j(x))^2)^{\frac{1}{2}} \nonumber \\
&\leq  2 \mathbb{E}_{\mu}(\Tr(\tilde{h}_j^2) - \Tr(\tilde{h}_j(x)^2))^{\frac{1}{2}}\| O \| \nonumber \\
&\leq 2(\eta_2(x))^{\frac{1}{2}} \| O \|,
\end{align}
where $\eta_2(x)=\mathbb{E}_{\mu }(\Tr(\tilde{h}_j^2) - \Tr(\tilde{h}_j(x)^2))$. From the first line to the second line, we used the sub-mulplicative property of the operator norm. From the second line to the third line, we used the well-known fact that the square of the operator norm is bounded by the Frobenius norm. From the third line to the fourth line, we used Eq.(\ref{eq:Frobenius_Identity}). In the last step, we used the concavity of the function $f(x)=x^{\frac{1}{2}}$ to complete the argument. While this bound is more crude compared Eq.(\ref{eq:eta_definition}), it is easier to deal with it numerically. If $\eta(x)$ decays exponentially, Definition 1 implies that the system is localized at some finite lengthscale. Indeed, we numerically observe a behavior that is consistent with the exponential decay at strong disorder $W\geq 4$, at which the model (\ref{eq:interacting_Hamiltonian}) is believed to be in the MBL phase; see FIG.\ref{fig:IoMDecay_interacting}.

\begin{figure}[h]
\includegraphics[width=3.3in]{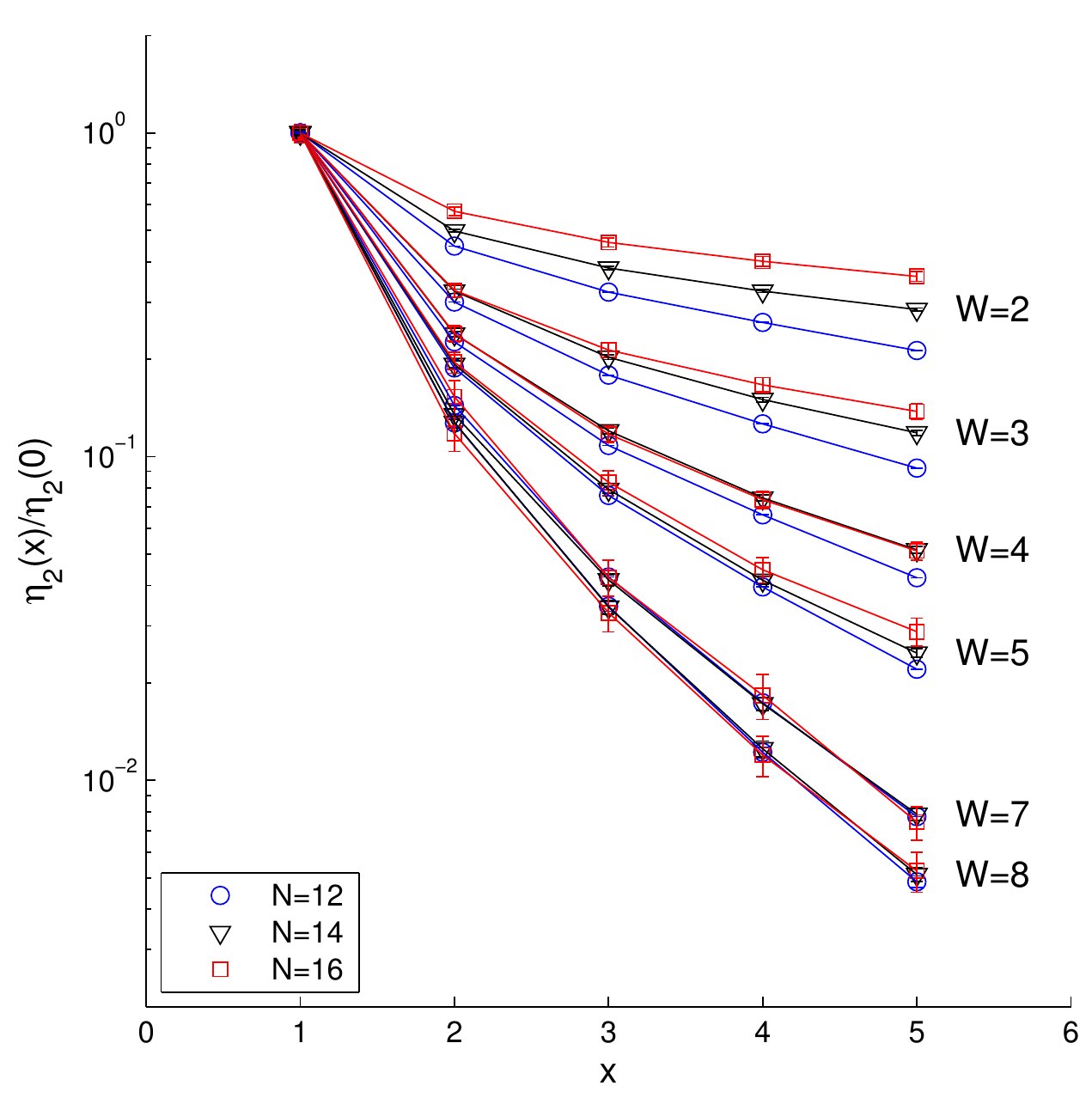}
\caption{Average decay property of the integrals of motion for interacting model(Eq.\ref{eq:interacting_Hamiltonian}). The system consists of $N=12,14,$ and $16$ spin-$\frac{1}{2}$ particles. The blue, black, and the red corresponds to the chain of length $12,14,$ and $16$. We studied the decay property of the integrals of motion $\tilde{h}_1$. The average is taken over $10^4$ disorder realizations for $N=12$ and 14, and $500$ realizations for $N=16$. For $W=4,5,7,8$, the decay rate is independent of the system size, which is indicative of the existence of finite localization length. The variation of $\eta_2(0)$ between different system sizes were within $\sim10\%$ of their average values. \label{fig:IoMDecay_interacting}}
\end{figure}

In addition to the consequences that we have briefly discussed in the introduction, our definition has a number of technically desirable properties. First, the definition does not involve any probability estimates, and as such, it is simple to state. Second, the definition does not involve the eigenstates. Any definition of localization length that involves an adiabatic connection to the noninteracting Anderson model is unlikely to be useful for a system that is topologically ordered and localized at the same time, e.g. disordered Majorana chain considered in Ref.[\onlinecite{Bravyi2011b}]. Our definition does not suffer from this problem.

\section{Logarithmic lightcone\label{section:ZVLR}}

In this Section, we study the main consequence of the exponential decay property of the integrals of motion: a variant of the Lieb-Robinson bound that exhibits a logarithmic lightcone. Our fundamental assumption is that $\tilde{h}_j$ decays exponentially on average in the MBL phase. In particular, we assume that the system is localized at a finite lengthscale $\xi$ which is independent of the system size; see Definition 1. While we are currently unable to prove this statement, FIG.\ref{fig:IoMDecay_noninteracting} and FIG.\ref{fig:IoMDecay_interacting} suggest that the assumption is valid for the models we have considered.

For a time evolution generated by a locally interacting Hamiltonian, Lieb-Robinson bound provides an upper bound on the correlation that is developed between two distant subsystems\cite{LR1972,Hastings2004a}:
\begin{equation}
\|[O_A(t), O_B]\| \leq c \min(|A|,|B|) e^{-a(x-vt)}, \label{eq:LRBound}
\end{equation}
where $O_A$ and $O_B$ are normalized operators that are supported on subsystem $A$ and $B$, $x$ is a distance between $A$ and $B$, $v$ is the Lieb-Robinson velocity, and $|A|$ is the volume of $A$.  We also employ a convention that $O_A(t)= e^{iHt}O_A e^{-iHt}$ is the operator in the Heisenberg picture.  We will also liberally use $c$ and $a$ as a numerical constant that does not depend on any of the aforementioned quantities.

Since Eq.(\ref{eq:LRBound}) is applicable to any locally interacting quantum many-body system, one cannot hope to derive the unique features of MBL from it. The main goal here is to strengthen Eq.(\ref{eq:LRBound}) by exploiting the decay property of $\tilde{h}_j$. To be more precise, we will prove a bound of the following form:
\begin{equation}
\mathbb{E}_{\mu}\| [O_A(t),O_B] \| \leq ct e^{-\frac{x}{2\xi}}. \label{eq:ZVLR}
\end{equation}

A similar bound was first proved in Ref.[\onlinecite{Burrell2007}] for disordered XY model, which can be mapped to a free-fermion model. Hamza et al. proved a much stronger bound for noninteracting systems that are localized; their bound does not have any time dependence, and they defined such bound as the zero-velocity Lieb-Robinson bound.\cite{Hamza2011} Intuitively, this bound implies that correlations do not spread at all for a spin model that can be mapped to the Anderson insulator. On the other hand, Eq.(\ref{eq:ZVLR}) allows correlations to spread, albeit very slowly. The physical meaning of Eq.(\ref{eq:ZVLR}) is clear; any system that is localized at a finite lengthscale exhibits an effective logarithmic lightcone.

For obvious reasons, the derivation of Eq.(\ref{eq:ZVLR}) is quite different from the derivation of the Lieb-Robinson bound. We do not wish to delve into these details in this paper, but we provide a side remark for the readers who are familiar with a modern derivation of the Lieb-Robinson bound.\cite{Hastings2004a} In the derivation, an upper bound of $\|[ O_A(t), O_B]\|$ is expressed in terms of $\|[h_j(t),O_B]\|$. A recursive inequality is subsequently derived from such a relation, which eventually leads to Eq.(\ref{eq:LRBound}). Our derivation differs in that we use a decomposition of $H$ into $\tilde{h}_j$. The main advantage of using $\tilde{h}_j$, as opposed to $h_j$, is that $\tilde{h}_j$ is invariant under the time evolution. This, in turn, implies that the recursion halts at the first level. While $\tilde{h}_j$ is not local, it can be well-approximated by a local operator if the system is localized at a finite lengthscale. An estimate on the approximation error can be related to the right hand side of Eq.(\ref{eq:ZVLR}). The precise technical result is stated below.
\begin{thm}
If a family of Hamiltonian over a disorder distribution $\mu$ is localized at a lengthscale of $\xi$,
\begin{equation}
\mathbb{E}_{\mu} \| [O_A(t),O_B] \| \leq ct |\partial A|e^{-\frac{x}{2\xi}},
\end{equation}
where $|\partial A|$ is the boundary area of $A$.
\end{thm}
The derivation shall be presented in the Appendix.

\section{Implications\label{section:Implications}}
In this Section, we study some of the implications of Theorem 1. Broadly speaking, we exploit the fact that the right hand side of Eq.(\ref{eq:ZVLR}) becomes $O(1)$ only after a time that scales exponentially in $x$. This fact can be straightforwardly used to show a strong suppression of transport. It is also an important ingredient behind the other results, such as the efficient simulability of the dynamics and the slow growth of average entanglement entropy under a quantum quench. While each of these results are written in a modular fashion, we believe it will be instructive to read them sequentially; we have deliberately coordinated these results with a sequence of increasing level of technical difficulties, so that a reader unfamiliar with such derivation can gradually get accustomed to the relevant techniques.

\subsection{Absence of transport}
As a first consequence, we give a strong yet indirect evidence that the system has a vanishing conductivity, independent of the choice of the conserved quantity. Intuitively, the observation that we are about to present here suggests that any particles must exhibit a subdiffusive behavior. In order to explain this phenomenon, imagine the following thought experiment. Prepare a system, say a spin chain, in one of its energy eigenstates. Perform a local physical operation, such as flipping a spin located at some site. Suppose we attempt to measure the distrubance that is caused by the physical operation by performing a measurement on a distant spin. In a delocalized regime, we expect to be able to detect this disturbance in a time that grows only polynomially in the distance. For example, if the spin chain shows a ballistic behavior, there must be a physical observable that detects such a disturbance in a time that is proportional to the distance between the two locations. We show that, for a MBL system, the time scales at least exponentially. Remarkably this conclusion is independent of the choice of the eigenstates, the physical operation, or the quantity which is being measured.

This result can be readily derived from Theorem 1, as we explain below. Without loss of generality, let us denote $\ket{\psi}$ as the eigenstate, $U_A$ as the unitary operator that represents the physical operation restricted on $A$, and $O_B$ as the observable that lies on $B$. The main quantity of interest is the time-dependent correlation function $\langle O_B(t)\rangle = \bra{\psi} U_A^{\dagger} e^{-iHt} O_B e^{iHt}U_A\ket{\psi}$. Using the fact that $\ket{\psi}$ is the eigenstate of the Hamiltonian, one can obtain the following identity:
\begin{equation}
\langle O_B(t)\rangle - \langle O_B(0)\rangle= \bra{\psi}U_A^{\dagger}(-t)[O_B, U_A(-t) ] \ket{\psi}. \nonumber
\end{equation}

Let us note two things: the expectation value of $U_A^{\dagger}(-t)[O_B,U_A(-t)]$ is bounded by $\|U_A^{\dagger}(-t)[O_B,U_A(-t)]\|$; and the left multiplication by $U_A^{\dagger}(-t)$ does not increase the norm. Therefore,
\begin{equation}
|\langle O_B(t)\rangle - \langle O_B(0) \rangle| \leq \|[O_B, U_A(-t)] \|. \nonumber
\end{equation}
Of course, the disorder average of the left hand side is bounded by the disorder average of the right hand side. Note that the commutator in the right hand side can be precisely related to the commutator that appears in Theorem 1. Hence, we conclude that
\begin{equation}
\mathbb{E}_{\mu} |\langle O_B(t)\rangle - \langle O_B(0) \rangle| \leq ct e^{-\frac{x}{2\xi}}, \label{eq:absence_transport}
\end{equation}
where $x$ is the distance between $A$ and $B$. In order to observe any disturbance in $A$, one needs to wait for at least  $t\approx e^{\frac{x}{2\xi}}$.

As we have stated before, our result suggests that any particles that can carry a conserved quantity must exhibit a subdiffusive behavior. This observation suggests that the conductivity must vanish for these systems. While we find such conclusion to be very likely, we are not aware of a proof that makes this intuition rigorous. We leave this as an open problem.

\subsection{Efficient simulation of the dynamics}
In quenched dynamics, time-dependent expectation value of a local observable is a frequnetly studied quantity. A related fundamental question lies on the computational cost for estimating such a quantity up to a desired precision.  For generic spin chains with local interactions, Lieb-Robinson bound has been used to estimate such a cost. The idea behind such analysis lies on the fact that there is an effective light cone so that the events outside the light cone can be ignored by sacrificing a modest amount of precision.  Some of the known algorithms exploiting this fact scale exponentially in $t$.\cite{Osborne2006,Hastings2008} For the MBL phase, we may expect to do better, since the effective speed of light is $0$.

Indeed, we show that there is an exponential speedup; typical instances of the disorder realizations can be simulated in a time that scales polynomially in $t$. The strategy is deceptively simple; approximate the time evolution as $e^{iH't}$, where $H'$ only contains terms that are nearby the support of $O$. This \emph{local propagator} acts very differently from the original one on the global state, but it preserves the time-dependent expectation value of $O$ approximately, as we explain below.

The fundamental relation that we are attempting to show can be summarized as follows:
\begin{equation}
\langle O(t)\rangle - \langle O'(t)\rangle \approx 0,\label{eq:approximation}
\end{equation}
where $O(t)=e^{iHt}Oe^{-iHt}$ and $O'(t)=e^{iH't} O e^{-iH't}$. Assuming one can efficiently obtain a reduced density matrix of a small-sized region, the above relation allows us to compute the time-dependent observable efficiently. If $H'$ acts on $n$ sites, $\bra{\psi}e^{iH't} O e^{-iH't} \ket{\psi}$ is equal to $\Tr(\rho e^{iH't} O e^{-iH't})$, where $\rho$ is the reduced density matrix over the support of $H'$. Such expectation value can be computed in a time that scales exponentially in $n$, but $n$ can be a number that is independent of the system size. Therefore, for a small value of $n$, one can expect to drastically reduce the computational cost. Intuitively, $\xi$ can be thought as an upper bound to the localization length, since the characteristic lengthscale at which the integrals of motion are spread out should be smaller than $\xi$.

Replacing the original propagator to the local propagator leads to the aforementioned computational advantage, but how much accuracy are we sacrificing? We show that not much is lost in this approximation; we obtain the following estimate on $n$ that is sufficient to approximate $\bra{\psi}e^{iHt} Oe^{-iHt} \ket{\psi}$ up to an accuracy of $\epsilon$:
\begin{equation}
n \approx \xi \log \frac{t^2}{\epsilon}. \label{eq:comp_cost}
\end{equation}
Such choice of $n$ amounts to a computational cost that scales as $\sim (\frac{t^2}{\epsilon})^{O(\xi)}$.

Eq.(\ref{eq:comp_cost}) follows from the same argument which was used in showing the absence of transport. While the algebraic manipulation is a bit more elaborate, we emphasize that there is a clear purpose behind it; namely, we are attempting to relate the left hand side of Eq.(\ref{eq:approximation}) to the commutator in Theorem 1.

There are three key logical steps behind this derivation. First, note that the difference between $\langle O(t)\rangle$ and $\langle O'(t)\rangle$ can be bounded by a more conservative estimate: the largest eigenvalue of $e^{iHt}Oe^{-iHt} - e^{iH't} O e^{-iH't}$. This is based on a simple observation that the difference can be recast as the expectation value of $e^{iHt}Oe^{-iHt} - e^{iH't} O e^{-iH't}$ for a state $\ket{\psi}$. By the definition of the norm, this is bounded by $\| e^{iHt}Oe^{-iHt} - e^{iH't} O e^{-iH't} \|$.

In the second step, we have the following sequence of identities:
\begin{widetext}
\begin{align}
\| e^{iHt}Oe^{-iHt} - e^{iH't} O e^{-iH't}\| &= \|U(t)^{\dagger}OU(t) - O \| \nonumber \\
&= \|\int^{t}_{0}\frac{d}{dt'} (U(t')^{\dagger}OU(t')) dt' \| \nonumber \\
&=\| \int^{t}_{0} iU(t')^{\dagger} [H''(t),O]U(t') dt' \|,
\end{align}
\end{widetext}
where $U(t)=e^{-iHt}e^{iH't}$ is some unitary operator and $H''= H - H'$ is a Hamiltonian whose support is separated from the support of $O$ by a distance of $\sim \frac{n}{2}$.\footnote{The factor of $\frac{1}{2}$ comes from the fact that roughly half of the $n$ sites lie on the left side of $O$ and the other half lie on the right side of $O$. Therefore, the distance between $O$ and the support of $H''$ is roughly $\frac{n}{2}$.} In the first line, we used the fact that the eigenvalue spectrum does not change under a conjugation by a unitary operator. In particular, the largest eigenvalue remains the same. In the second line, we have rewritten the expression using an auxiliary variable $t'$. The last line follows by simply calculating the derivative of $U(t')$. The main purpose of the second step was to obtain a commutator $[H''(t), O]$, whose norm can be bounded from Theorem 1.

Indeed, as the last step, one can use the triangle inequality to show that the last line is bounded by $\int^{t}_0\|[H''(t),O]\| dt' $. Taking the expectation value, we obtain the following universal bound:
\begin{equation}
\mathbb{E}_{\mu}|\langle O(t)\rangle - \langle O'(t)\rangle| \leq \int^{t}_{0}  \mathbb{E}_{\mu}\|[H''(t'),O]\| dt'.
\end{equation}
Applying Theorem 1 to each of the local terms in $H''$, and performing the integral, we conclude that
\begin{equation}
\mathbb{E}_{\mu}|\langle O(t)\rangle - \langle O'(t)\rangle| \leq ct^2 e^{-\frac{n}{4\xi}}.\label{eq:efficient_simulation_errorbound}
\end{equation}
In order to obtain an accuracy of $\epsilon$ for any local operator $O$, it suffices to choose $n$ as $4\xi \log(\frac{ct^2}{\epsilon})$.

\subsection{Slow growth of entanglement entropy}
One of the first clear signatures of the MBL phase was the slow growth of the entanglement entropy under a quantum quench. For example, average entanglement entropy across a cut in the middle of a spin chain was studied in Ref.[\onlinecite{Prosen2008,Bardarson2012,Serbyn2013-1,Calabrese2005,Calabrese2006,Kim2013a}]. A logarithmic growth in time was observed in the MBL phase, while a linear growth was observed in the delocalized phase. Here we derive a bound that closely matches the behavior of the MBL phase, assuming the system is localized at some finite lengthscale:
\begin{equation}
\mathbb{E}_{\mu} S(t) \leq c\xi \log t + o(1),
\end{equation}
where $\mathbb{E}_{\mu}S(t)$ is the disorder-averaged entanglement entropy at time $t$ and $o(1)$ is a term that approaches $0$ in the $t\to \infty$ limit. The initial state is assumed to be a product state. The derivation is admittedly more involved, as it is based on two ideas that have received little attention in the studies of localization so far. Instead of explaining these ideas right away, we would like to begin by explaining why such tools are necessary.

Recall that, at an abstract level, we are attempting to bound the growth of the entanglement entropy of some finite region under a unitary dynamics generated by a sum of local Hamiltonians, e.g., Eq.(\ref{eq:AndersonModel}). The solution to this problem for \emph{general} Hamiltonians is well-known; the rate at which entanglement entropy increases is bounded as follows:
\begin{equation}
\frac{dS}{dt} \leq c\sum_{i \in \partial}\| h_{i}\| \log d_i, \label{eq:entangling_rate}
\end{equation}
where the sum is taken over the terms that act across the boundary between the region and its complement and $d_i$ is the dimension of the support of $h_i$.\cite{Acoleyen2013} This bound is an upper bound on the \emph{entangling rate}: the rate at which entanglement entropy increases. This upper bound, applied to our systems of interest, implies that the entanglement entropy of a subsystem $A$ in a time interval of  $t$ can grow at most by $\sim c|\partial A|t$. This is known to be optimal in a sense that for some $c<c_0$, there exists some Hamiltonian such that its entangling rate exceeds the bound. It should be also noted that such a linearly growing upper bound is obtained using other techniques as well; see Ref.[\onlinecite{Eisert2006}].

We will end up using the upper bound on the entangling rate, but one thing is clear from the above analysis: that na\"ively applying it to our systems of interest cannot yield the desired logarithmically increasing bound. Indeed this was to be expected, since the upper bound on the entangling rate is applicable to both ergodic and localized systems. How then can we use Theorem 1 to attain a more refined upper bound?

We can do that by considering a particular decomposition of the propagator; our approach is to decompose the propagator $U(t)=e^{iHt}$ as follows:
\begin{equation}
U(t) = U_{in}(t) U_{out}(t) U_{\delta}(t) U_{\partial}(t),\label{eq:decomposition}
\end{equation}
where $U_{in}(t)$ is a unitary operator that acts only inside a given subsystem, $U_{out}(t)$ is  a unitary operator that acts only outside the given subsystem, and $U_{\partial}(t)$ is a unitary operator that acts in the vicnity of the boundary between the left and the right part. Once these operators are specified, $U_{\delta}(t)$ shall be defined implicitly from them. It is important to note that each of these unitary operators play a different role in creating entanglement. Neither $U_{in}(t)$ nor $U_{out}(t)$  can create any entanglement, since they are only able to unitarily rotate the state locally. As far as the entaglement is concerned, $U_{\delta}(t)$ and $U_{\partial}(t)$ are the only relevant operators. Our goal is to argue that the \emph{entangling power} of these unitary operators are not so large; they cannot create too much entanglement.

There are simple reasons behind why these operators($U_{\delta}(t)$ and $U_{\partial}(t)$) can only create a small amount of entanglement. The case for the $U_{\partial}(t)$ is particularly simple, as it only acts nontrivially on the finite number of particles that are nearby the boundary between the left and the right part. Thus the amount of entanglement that can be created by this operator is bounded by some constant. The remaining operator, $U_{\delta}(t)$ acts nonlocally, but we will be able to exploit Theorem 1 to show that it can be generated by a sum of quasi-local terms whose strength is sufficiently small. A possible increase in the entanglement entropy that is due to $U_{\delta}(t)$ will be bounded by exploiting the bound on the entangling rate.

In the above analysis, we have only mentioned the roles of the unitary operators that appear in Eq.(\ref{eq:decomposition}). Now it is time to specify them, for a model that can be described by the following Hamiltonian:
\begin{equation}
H= H_{in} + H_{out} + H_{\partial},
\end{equation}
where $H_{in}$($H_{out}$) consists of terms that act  only on the inside(outside) of a given subsystem and $H_{\partial}$ consists of terms that act nontrivially on both the left and the right part. For example, for a one-dimensional system with a cut in the middle, $H_{\partial}$ consists of a finite number of terms. Without loss of generality, we shall define $H_{\partial}$ as follows:
\begin{equation}
H_{\partial} = \sum_{n \in \partial} h_{n}, \nonumber
\end{equation}
where the summation is taken over the local terms in the Hamiltonian that act nontrivially both on the inside and the outside of the given region.

We choose the unitary operators in Eq.\ref{eq:decomposition} in such a way that they satisfy the following differential equations.
\begin{align}
&\frac{d}{dt}U_{in}(t) = iH_{in} U_{in}(t) \nonumber\\
&\frac{d}{dt}U_{out}(t) = iH_{out} U_{out}(t) \nonumber\\
&\frac{d}{dt}U_{\partial}(t) = i \sum_{n\in \partial} [h_{n}(t)]_r U_{\partial}(t),
\end{align}
where $h_{n}(t) = e^{iHt} h_{n}e^{-iHt}$ and $[h_{n}(t)]_r$ is the restriction of $h_{n}(t)$ to a set of sites which are distance $r$ or less away from the support of $h_{n}$, which is defined below.
\begin{equation}
[O]_r = \frac{I_{\overbar{A(r)}}}{d_{\overbar{A(r)}}}\otimes\Tr_{\overbar{A(r)}} O, \nonumber
\end{equation}
where $A(r)$ is a set of sites which are distance $r$ or less away from the support of $O$ and $\overbar{A(r)}$ is its complement. This region appears in the subscripts of the symbols which represent the dimension($d$) and the identity operator acting on those sites($I$), as well as the partial trace($\Tr$) over the region. It is worth noting that $[O]_r$ acts nontrivially only on $A(r)$, and approaches $O$ as $r$ approaches infinity. It should be also noted that $r$ is a parameter that shall be optimized to obtain the best bound.

Beginning from a product state, the entanglement entropy at time $t$, which is denoted as $S(t)$, must obey the following upper bound:
\begin{widetext}
\begin{equation}
 S(t) \leq cr + c'\sum_{n\in\partial}\sum_{r'=r}^{\infty} \int^t_0 (r'+r+1)\| [h_{n}(t')]_{r'+1} - [h_{n}(t')]_{r'}\| dt'.\label{eq:entropy_upperbound}
\end{equation}
\end{widetext}
The above inequality follows from two observations. First, the difference between the entanglement entropy of the initial state $\ket{\psi}$ and $U_{\partial}(t)\ket{\psi}$ is bounded by a number of sites for which $U_{\partial}(t)$ acts nontrivially; this contribution corresponds to the first term in the upper bound. Second, the difference between the entanglement entropy of $U_{\partial}(t)\ket{\psi}$ and $U_{\delta}(t)U_{\partial}(t)\ket{\psi}$ can be bounded from the bound on the entangling rate. To see this, let us write down the differential equation that governs the evolution of $U_{\delta}(t)$:
\begin{equation}
\frac{d}{dt}U_{\delta}(t) = iU_{\delta}(t) \sum_{n\in \partial}U_{\partial}(t) (h_{n}(t) - [h_{n}(t)]_r)U_{\partial}(t)^{-1}.
\end{equation}
The ``Hamiltonian'' that generates this flow is $\sum_{n\in \partial}U_{\partial}(t) (h_{n}(t) - [h_{n}(t)]_r)U_{\partial}(t)^{-1}$, which can be expanded as follows:
\begin{equation}
\sum_{n\in \partial}\sum_{r'=r}^{\infty}U_{\partial}(t) ([h_{n}(t)]_{r'+1} - [h_{n}(t)]_{r'})U_{\partial}(t)^{-1}.
\end{equation}
For each $r'$ and $n$, we can choose $U_{\partial}(t) ([h_{n}(t)]_{r'+1} - [h_{n}(t)]_{r'})U_{\partial}(t)^{-1}$ to be the $h_i$ which appears in Eq.\ref{eq:entangling_rate}. The operator norm of this term is identical to $[h_{\partial}(t)]_{r'+1} - [h_{\partial}(t)]_{r'}$, since the norm is invariant under a unitary rotation. The size of their support is bounded by $2(r'+r+1)$, since (i) $[h_{\partial}(t)]_{r'+1} - [h_{\partial}(t)]_{r'}$ is supported on at most $2(r'+1)$ sites and (ii) $U_{\partial}(t)$ is supported on a set of sites that are distance $r$ or less away from the boundary. Applying these observations to Eq.\ref{eq:entangling_rate}, one can obtain the second term that appears in the upper bound of Eq.\ref{eq:entropy_upperbound}.

By performing the disorder average on Eq.\ref{eq:entropy_upperbound}, we obtain the following bound:
\begin{widetext}
\begin{equation}
 \mathbb{E}_{\mu}S(t) \leq cr + c' \sum_{n\in \partial}\sum_{r'=r}^{\infty} \int^t_0 (r'+r+1) \mathbb{E}_{\mu}\| [h_{n}(t')]_{r'+1} - [h_{n}(t')]_{r'}\| dt'.\label{eq:entropy_upperbound_average}
\end{equation}
\end{widetext}
The following lemma provides an upper bound on $\mathbb{E}_{\mu}\| [h_{n}(t')]_{r'+1} - [h_{n}(t')]_{r'}\|$:
\begin{lem}
\begin{equation}
\mathbb{E}_{\mu}\| [h_{n}(t')]_{r'+1} - [h_{n}(t')]_{r'}\|  \leq ct'e^{-r'/\xi}
\end{equation}
\end{lem}
\begin{proof}
Note that
\begin{equation}
[h_{n}(t')]_{r'} = \int d\mu(U) U h_{n}(t') U^{\dagger}, \nonumber
\end{equation}
where $\mu(U)$ is a Haar measure over a set of unitary operators acting on the complement of the support of $[h_{n}(t')]_{r'}$.\cite{Bravyi2006} For such $U$,
\begin{equation}
\mathbb{E}_{\mu}\| Uh_n(t')U^{\dagger} - h_n(t')\|=\mathbb{E}_{\mu}\| [U, h_{n}(t')\| \leq ct' e^{-r'/\xi}. \nonumber
\end{equation}
Integrating over $\mu(U)$ and using the triangle inequality, we conclude that
\begin{equation}
\mathbb{E}_{\mu}\|[h_n(t')]_{r'}  - h_n(t')\| \leq ct' e^{-r'/\xi}. \nonumber
\end{equation}
A similar inequality can be derived for $[h_n(t')]_{r'+1}$ as well, after which one can bound the closeness between $[h_n(t')]_{r'+1}$ and $[h_n(t')]_{r'}$.
\end{proof}
The above Lemma immediately implies that the disorder-averaged entanglement entropy obeys the following universal bound:
\begin{equation}
\mathbb{E}_{\mu}S(t) \leq cr +c'rt^2e^{-r/\xi}. \nonumber
\end{equation}
By choosing $r=O(\xi \log t)$
\begin{equation}
\mathbb{E}_{\mu}S(t) \leq c\xi \log t + o(1),
\end{equation}
where $c$ is some constant and $o(1)$ is a term that vanishes in $t\to \infty$ limit. Of course, for short times, one can simply use the bound on the entangling rate, which implies that
\begin{equation}
\mathbb{E}_{\mu}S(t)\leq c t.
\end{equation}
Since both bounds are applicable to our system of interest, one can simply take the minimum, which is achieved by the linearly increasing bound for short times and logarithmically increasing bound for long times.

\section{Discussion}
Starting from a generic decay property of the integrals of motion defined in Ref.[\onlinecite{Chandran2014}], we were able to derive a number of qualitative features of MBL. We emphasize once more that our conclusion is independent of the details of the model, so long as the system is localized at a finite lengthscale, a notion that we made precise in Definition 1.

We emphasize that our bound on information propagation, i.e., Theorem 1, is valid in any dimensions. Since the implications that were discussed in the context of one-dimensional systems were primarily based on Theorem 1, we believe that most of the arguments can be straightforwardly generalized. In particular, Eq.\ref{eq:absence_transport} can be proved without modifying the current proof. We believe these general facts warrant a further study on these integrals of motion.

Clearly, an important open question at this point is whether one can prove that a sufficiently strongly disordered system is localized at a finite lengthscale. Another interesting open question is whether the integrals of motion defined in this paper remain local in systems with a robust gapless edge mode. By definition, the integrals of motion cannot be localized everywhere; it would imply that the system obeys a bound like Eq.(\ref{eq:ZVLR}), which would be inconsistent with the (thermal) transport along the edge. One logical possibility is that the integrals of motion are localized in the bulk, but not near the edge. It will also be interesting to bound the entanglement entropy under a quantum quench in higher-dimensional systems. We conjecture that entanglement grows as $|\partial A| \log t$ at large $t$ in the MBL phase, where $|\partial A|$ is the area of the entanglement cut.

In the context of studying the quenched dynamics of a microscopic model, Eq.(\ref{eq:efficient_simulation_errorbound}) suggests that most of the instances of a disordered system can be efficiently simulated. This does not directly imply that we can unconditionally simulate such systems efficiently, since we do not have an efficient method to verify that the system satisfies Definition 1. However, there might be an efficiently checkable condition, under which a truncation like Eq.(\ref{eq:comp_cost}) can be justified. We leave that as an open problem.

The approach that we have taken in this paper is primarily motivated by MBL. However, we believe the approach will prove useful in studying the dynamics of other types of  interacting quantum many-body systems as well, for the reasons that we explain below. The only postulate that we have imposed on our system is the average spatial decay property of the integrals of motion which are canonically defined for any locally interacting Hamiltonian; see Definition 1. We expect to observe a different decay behavior for other systems, which would imply a different conclusion. For example, some systems might have a gapless edge mode, and hence exhibit a different decay behavior near the boundary. Some systems might support ballistic or diffusive transport, which would again imply a different decay behavior. We leave these studies for future work.

Lastly, we note that our work was largely inspired by the fractional moment method for studying Anderson localization, which uses the average decay property of the fractional moment of the Green's function to prove dynamical localization.\cite{Aizenman1993} The role of the Green's function is taken by the integrals of motion in our work, and its strength was numerically shown to be decaying exponentially on average. Similar to the fractional moment method, the average decay property of the integrals of motion can be exploited to prove statements that are reminiscent to the dynamical localization. This connection remains as an analogy in this primordial form, but it would be interesting to understand if there are any insights from the fractional moment method that can prove useful in our approach.

\bibliography{bib}

\appendix
\section{Derivation of Theorem 1}
Theorem 1 follows from Lemma 2, which is stated below. Let us first begin by defining some notations.
We denote $A(r)$ as a set of sites which are distance $r$ or less away from $A$. We define $A(r)^c$ as a complement of $A(r)$, i.e., a set of sites whose distance from $A$ are larger than $r$. Note that, without loss of generality, the Hamiltonian can be written as follows:
\begin{equation}
H=\tilde{H}_{A(r)} + \tilde{H}_{A(r)^c},
\end{equation}
where
\begin{equation}
\tilde{H}_{A(r)}=\sum_{n\in A(r)} \tilde{h}_j,
\end{equation}
and the other term is defined similarly.

\begin{widetext}
\begin{lem}
\begin{equation}
\|[O_A(t),O_B] \| \leq  2t(\|O_B \|\|[O_A, \tilde{H}_{A(r)^c}] \| + \| O_A \| \|[O_B,\tilde{H}_{A(r)}] \|)
\end{equation}
\end{lem}
\end{widetext}
\begin{proof}
First, define $f(t)= [O_A(t),O_B]$.
\begin{align}
f'(t)&=i[e^{iHt}[H,O_A]e^{-iHt},O_B] \nonumber \\
&=i[e^{iHt}[\tilde{H}_{A(r)} + \tilde{H}_{A(r)^c}, O_A]e^{-iHt},O_B] \nonumber \\
&=g_{A(r)}(t) + g_{A(r)^c}(t),
\end{align}
where
\begin{align}
g_{A(r)}(t)&=i[e^{iHt}[\tilde{H}_{A(r)} , O_A]e^{-iHt},O_B] \nonumber \\
g_{A(r)^c}(t)&=i[e^{iHt}[\tilde{H}_{A(r)^c} , O_A]e^{-iHt},O_B].
\end{align}
Further,
\begin{align}
g_{A(r)}(t)&=i[[\tilde{H}_{A(r)}, O_A(t)],O_B] \nonumber \\
&=-i[f(t),\tilde{H}_{A(r)}] -i [[O_B,\tilde{H}_{A(r)}],O_A(t)].
\end{align}
Combining these identities together, $f(t)$ satisfies the following differential equation:
\begin{equation}
f'(t)=-i[f(t),\tilde{H}_{A(r)}]+ \delta(t),\label{eq:DifferentialEquation}
\end{equation}
where
\begin{equation}
\delta(t)= i[e^{iHt}[\tilde{H}_{A(r)^c} , O_A]e^{-iHt},O_B]-i [[O_B,\tilde{H}_{A(r)}],O_A(t)].
\end{equation}
Since the first term in Eq.\ref{eq:DifferentialEquation} is norm-preserving,
\begin{equation}
\| f(t) \| \leq \| f(0)\| + \int^t_0 \| \delta(t') \| dt'.
\end{equation}
Using the unitarity of $e^{iHt}$ and triangle inequality, $\|\delta(t')\|$ is uniformly bounded as follows:
\begin{equation}
\|\delta(t')\| \leq 2(\|O_B \|\|[O_A, \tilde{H}_{A(r)^c}] \| + \| O_A \| \|[O_B,\tilde{H}_{A(r)}] \|).
\end{equation}
Integrating out $t'$, the lemma is proved.
\end{proof}

Now that we have Lemma 1, the derivation of Theorem 1 is straightforward. One can simply take the expectation value over the disorder realizations, and choose $r=\frac{x}{2}$. Since the distance between $A(\frac{x}{2})$ and $B$ as well as the distance between $A(\frac{x}{2})^c$ and $A$ are at least $\frac{x}{2}$, the relevant commutators decay exponentially in $x$. Theorem 1 follows by summing all the contributions.

\end{document}